\documentclass[a4paper,USenglish]{lipics-v2021}

\pdfoutput=1 
\hideLIPIcs  

\title{Exact Competitive Ratios of Index Rules for Scheduling with Obligatory Testing} 

\author{Felix {Buld}}{Operations Research, Technical University of Munich, Arcisstr. 21, 80333 Munich, Germany}{felix.buld@tum.de}{https://orcid.org/0009-0007-0230-8024}{Partly funded by the German Research Foundation (DFG) - GRK 2201/2 - Project number 277991500.}
\authorrunning{F. Buld}

\Copyright{Felix Buld}



\nolinenumbers 
\usepackage{nicefrac}
\newcommand{\E}[0]{\mathbb{E}}

\begin{document}

\maketitle

\begin{abstract}
We revisit scheduling with obligatory testing, where each job consists of two operations: a test with an initially known time and an operation whose processing time is revealed upon completion of the job's test. We develop an iterative competitive-ratio-preserving instance-reduction method that enables us to compute exact competitive ratios directly on highly structured instances for various index rules. Our main result is a clean analysis showing that the competitive ratio of the algorithm $1$-SORT is exactly the Golden Ratio $\varphi\approx 1.618$, improving the recent best-known upper bound of $1.861$ for arbitrary testing times. We further leverage this method to obtain near-tight bounds for a simple randomized algorithm on instances with uniform testing times.
\end{abstract}

\section{Introduction}
This paper studies scheduling under \emph{explorable uncertainty}, where some job parameters are initially unknown but can be revealed through time-consuming tests. This uncertainty reduction comes with a trade-off towards the completion of jobs for which full information is already available. Thus, the model bridges classical offline and online scheduling settings, augmented with testing as a layer of active information acquisition at the cost of time investment. The paper follows a line of research on scheduling with testing that originated in the foundational works of \cite{durr2020adversarial,levi2019scheduling} and provides a competitive analysis that characterizes the exact performance guarantee of the studied scheduling rules.

\subsection{Problem Definition}\label{sec: Problem Definition}
In this paper, we study scheduling with obligatory testing as introduced by Dogeas, Erlebach, and Liang~\cite{dogeas2024obligatory}. In this model, $n$ jobs are to be scheduled on a single machine. Each job $j=1,\ldots,n$ consists of two operations: an obligatory test of testing time $t_j\geq 0$, and a processing operation of processing time $p_j\geq 0$ that can start only after the test, at any later time. The processing time of job~$j$ is revealed only after its test is completed; the testing time is known in advance. These $2n$ operations are to be scheduled non-preemptively, one operation at a time, on a single machine. Let $C_j$ denote the completion time of job~$j$ in a schedule. The objective is to minimize the total completion time $\sum C_j$. The goal of this paper is to devise and analyze deterministic and randomized algorithms of provably good performance.

\subsection{Preliminaries and Related Work}
Let $\Psi$ be the set of instances of this scheduling problem. For an instance $I\in\Psi$ of this scheduling problem, let $ALG(I)$ denote the algorithmic cost and $OPT(I)$ the optimal cost, or $ALG$ and $OPT$ if the instance is clear from the context. The performance is measured by the competitive ratio, i.e., $\rho(ALG):=\sup_{I\in\Psi} \nicefrac{ALG(I)}{OPT(I)}$ for a deterministic algorithm and $\rho(ALG):=\sup_{I\in\Psi} \nicefrac{E\left[ALG(I)\right]}{OPT(I)}$ for a randomized algorithm, against an oblivious adversary. 
In the optimal (offline) solution with full information, jobs are sorted by \emph{shortest processing time first} (\textsc{SPT}), i.e., non-decreasingly in their total processing time $t_j + p_j$, and scheduled in this order, without preemption between the two operations of each job. Let $\sigma$ denote the permutation that sorts jobs in this \textsc{SPT}-order. Then, the optimal cost is
\begin{equation}\label{eq: Optimal cost}
    OPT=\sum_j (n+1-j)(t_{\sigma(j)}+p_{\sigma(j)})=\sum_j (t_j + p_j) + \sum_{i<j} \min\lbrace t_i+p_i, t_j+p_j\rbrace \ .
\end{equation}
The second expression is a reformulation in terms of \emph{pairwise delays}. In general, with~$d(i,j)$ denoting the \emph{delay} of job $j$ caused by job $i$, i.e., the time spent on testing or processing job $i$ before $j$ has been completed, here $d(i,j)\in\lbrace 0, t_i, t_i + p_i\rbrace$, the algorithmic cost can be written as
\begin{equation}\label{eq: Algorithmic cost}
    ALG=\sum_j (t_j + p_j) + \sum_{i<j} d(i,j) + d(j,i) \ . 
\end{equation}
This pairwise-delay decomposition is often used to succinctly represent the algorithmic cost of algorithms and analyze their performance, and also extends to randomized algorithms using (expected) algorithmic cost, see \cite{buld2025scheduling,cosson2026online,dogeas2024obligatory,liu2023power,motwani1994nonclairvoyant}.

In the classical non-clairvoyant scheduling problem, with the objective of minimizing~$\sum C_j$, the processing time $p_j$ is revealed only upon completion of a job; obligatory tests are not available. The non-preemptive variant of this scheduling problem is unbounded, whereas in the preemptive version, the \emph{Round Robin}-Rule, which rotates among infinitesimal pieces of incomplete jobs, achieves the best-possible competitive ratio of $2$ \cite{motwani1994nonclairvoyant}. The algorithm and its analysis carry over to obligatory testing if one allows arbitrary preemption, but the testing model provides strictly more information and can, without loss of generality, be considered as the \emph{test-preemptive} model in Section~\ref{sec: Problem Definition} \cite{buld2025scheduling,dogeas2024obligatory,durr2020adversarial,liu2023power}, thereby allowing strictly smaller competitive ratios for algorithms.

The obligatory testing model in Section~\ref{sec: Problem Definition} was first studied by Dogeas et al.~\cite{dogeas2024obligatory}. They consider the rule $\beta$-SORT, where testing operations have priority value $\beta t_j$ and processing operations have priority value $ p_j$, and prove that $1$-SORT has a competitive ratio of at most $1.861$ using an enhanced amortization scheme based on the pairwise-delay framework. They also provide a lower bound based on the parameter $\beta>0$. This lower bound is minimal for $1$-SORT and equals the \emph{Golden Ratio} $\varphi:=\nicefrac{(1+\sqrt{5})}{2}\approx 1.618$. This is the positive root of the quadratic equation $x^2-x-1=0$ and a recurring parameter in the competitive analysis of various algorithms in the scheduling-with-testing literature (see, e.g., \cite{durr2020adversarial}) and beyond.

For uniform testing times, they use an instance-reduction framework from Dürr, Erlebach, Megow, and Meißner~\cite{durr2020adversarial} to compute the exact competitive ratio for SIDLE, a variant of $\beta$-SORT on these instances, which is minimized for $\beta\approx 1.355$, and equal to about $1.585$. They also give a lower bound of $\sqrt{2}$ for uniform testing times. This bound was improved by recent papers \cite{liang2026hardness,rozhon2026test}, which give lower bounds of $\nicefrac{3}{2}$ \cite{liang2026hardness} and about $1.57$ \cite{rozhon2026test} for any deterministic algorithm, and of $\nicefrac{4}{3}$ \cite{rozhon2026test} for any randomized algorithm. For the latter two results, algorithms with matching bounds of rather intricate structure and analysis were developed with substantial assistance from large language models~\cite{rozhon2026test}.
Results for index rules for scheduling on a single machine can be transferred to parallel identical machines, potentially with some performance loss, using the frameworks developed in \cite{buld2025scheduling,eastman1964bounds,gong2024approximation,liang2026hardness}.
The obligatory testing model and the algorithm $1$-SORT were also generalized to jobs arriving online and consisting of chains of successive operations that are gradually revealed, and a competitive ratio of $2$ was achieved in the obligatory testing model for instances with uniform testing times for total flow time minimization \cite{lindermayr2026gradually}.

A very related model from which the obligatory testing model originated is the \emph{optional} testing model. In this model, one is given intervals for processing times, i.e., $p_j\in\left[0,u_j\right]$ for each job~$j$, and can either choose to process a job untested for $u_j$ units of time or test it, taking time $t_j$, and process it later for $p_j$ units of time. Similar index rules to $\beta$-SORT and generalizations thereof for optional testing have been studied in various papers (e.g., \cite{albers2021explorable,buld2025scheduling,durr2020adversarial,gong2024approximation,liu2023power}). The rule $\beta$-SORT for obligatory testing originated from an extended rule $(\alpha,\beta)$-SORT for optional testing \cite{albers2021explorable,dogeas2024obligatory}. Instance-reduction frameworks have been used for algorithms in the model with uniform testing times \cite{durr2020adversarial}. For the general model, the analysis techniques have mostly focused on locally bounding pairwise delays \cite{buld2025scheduling,gong2024approximation,liu2023power}.

In the stochastic counterparts for scheduling with testing \cite{levi2019scheduling,levi2024scheduling,sun2018patient}, processing times and job priorities are drawn from independent probability distributions. For obligatory testing in this context, Levi, Magnanti, and Shaposhnik~\cite{levi2024scheduling} prove the optimality of an index rule.

Approximation-ratio-preserving instance-reduction approaches have been used for scheduling rules in other contexts, e.g., \cite{schwiegelshohn2011alternative}, or for integrality-gap-preserving reductions within linear programming-based approaches to other combinatorial optimization problems \cite{encz2026integrality}.

\subsection{Contributions}
Our main contributions in this paper are threefold. 

First, we develop an instance-reduction method that is competitive-ratio-preserving, i.e., a sequence of structural adaptations, such as shifts in processing times and adversarial exchange arguments, applied to an instance without decreasing an algorithm's competitive ratio. This process results in a highly structured family of instances for which the exact competitive ratio can be computed directly. 

Second, we apply this method to the algorithmic family $\beta$-SORT and characterize its exact competitive ratios as functions of $\beta$. This yields an improvement for the competitive-ratio for $1$-SORT from the currently best-known guarantee $1.861$ \cite{dogeas2024obligatory} to the Golden Ratio $\varphi\approx 1.618$ and solves the open question on the exact competitive ratio by Dogeas et al.~\cite{dogeas2024obligatory}. It also shows that a simple index rule almost matches the known lower bounds. We further prove that deeper amortization, beyond locally bounding pairwise costs, is necessary to achieve exact competitive ratios for this algorithmic family. Our method does this via a clean analysis. 

Third, we connect the two (previously separate) lines of research on stochastic and adversarial scheduling with testing via Yao's Principle. We leverage this to provide lower bounds on randomized algorithms for adversarial scheduling with obligatory testing of $\nicefrac{4}{3}$ in the uniform testing case. We provide a family of randomized algorithms that are simple variants of the deterministic family $\beta$-SORT paired with a random permutation of tests. Using our instance-reduction method, we can show that, with an optimized parameter $\beta \approx 1.947$, this achieves an almost-optimal competitive ratio of at most $1.359$ in the uniform testing case.

\subsection{Outline of the Paper}
This paper is structured as follows. We start by presenting our instance-reduction method in Section~\ref{sec: Instance-Reduction}. In Section~\ref{sec: Deterministic Algorithms}, we present the family of deterministic algorithms $\beta$-SORT and apply our instance-reduction method to characterize the exact competitive ratios. In Section~\ref{sec: Randomized Algorithms}, we study randomized algorithms and also employ our instance-reduction method together with a presentation of a lower bound via the connection between stochastic and adversarial scheduling with testing. Finally, a conclusion is drawn in Section~\ref{sec: Conclusion}.

\section{A Competitive-Ratio-Preserving Instance-Reduction Method}\label{sec: Instance-Reduction}
In this section, we first describe the high-level idea behind our instance-reduction method, followed by the technical aspects used throughout this paper to exactly compute the competitive ratios of various algorithms. The method applied to an arbitrary instance of the scheduling problem consists of three main steps:
\begin{quote}
\begin{description}
    \item[1. Competitive-Ratio-Preserving Shifts] Iteratively, we jointly increase or decrease a subset of testing or processing times to strictly decrease the number of distinct parameters in an instance.
    \item[2. Adversarial Exchange Arguments] Tie-breaks in the algorithmic structure are resolved in favor of the adversary without changing the structure of the optimal solution. 
    \item[3. Exact Computation on Highly Structured Instances] For worst-case instances with a small number of distinct parameters, algorithmic and optimal costs are highly structured, and the exact competitive ratios can directly be computed.
\end{description}
\end{quote}

The first step of the instance-reduction relies heavily on the following structural properties in Lemma~\ref{lemma: maximum at the bounds} and~\ref{lemma: shift to the bounds}. These two results are generalizations of two results in \cite[Proposition $2$ and Lemma $3$]{durr2020adversarial}, which provided a similar reasoning for optional testing and uniform testing times that were applied in \cite{dogeas2024obligatory,durr2020adversarial}. The main insight behind the enhanced methods in the first step is that we need a more general property of the optimal cost and iterative shifts to make these results applicable to non-uniform testing times. Lemma~\ref{lemma: maximum at the bounds} is the technical core of this step.

\begin{lemma}\label{lemma: maximum at the bounds}
    Consider a parameter $x\in\left[l,u\right]$ for bounds $0\leq l<u$. Suppose that on this interval, $f(x)=a+bx$ with $a,b\geq0$, and $g$ is a piecewise affine, continuous, non-decreasing, concave function with $g(l)>0$. Then, $\max_{x\in\left[l,u\right]} \nicefrac{f(x)}{g(x)}$ is attained for at least one value $x^*\in\lbrace l, u\rbrace$.
\end{lemma}
\begin{proof}
    We write $g(x)=c_k + d_kx$ for all $x\in \left[x_k,x_{k+1}\right]$ for a finite partition $l=x_0 < x_1 < \ldots < x_K=u$ and all $k=0,\ldots,K-1$ with $d_k > d_{k+1}$ and $c_k < c_{k+1}$ using the stated properties of $g$ and the interval.
    On the interior of each sub-interval $\left(x_k,x_{k+1}\right)$, the derivative of $\nicefrac{f}{g}$ equals
    $\nicefrac{(bc_k-ad_k)}{g(x)^2}$. The sign of the derivative is independent of $x$. Since $a,b\geq0$, $d_k- d_{k+1}>0$ and $c_{k+1}-c_k > 0$, we have $bc_k-ad_k \leq bc_{k+1}-ad_{k+1}$ for each~$k=0,\ldots,K-2$, implying that the sign is non-decreasing in $k$. Hence, we can conclude that the ratio is first non-increasing, then non-decreasing, and attains its maximum at least at one of the bounds.
\end{proof}

Consider an arbitrary but fixed instance $I$ of the scheduling problem with obligatory testing and an arbitrary algorithm $ALG$. 
Lemma~\ref{lemma: maximum at the bounds} will now be applied to the cost functions in this setting with $f=ALG$ and $g=OPT$ to obtain another instance parametrized by some value $x$, without decreasing the ratio. For this, we may write testing and processing times as $t_j(x)$ and $p_j(x)$, the algorithmic cost as $ALG(x)$, the optimal cost as~$OPT(x)$, and the pairwise delay as $d_x(i,j) + d_x(j,i)$ as functions in $x$.

\begin{lemma}\label{lemma: shift to the bounds}
         Consider a parameter $x\in\left[l,u\right]$ for bounds $0\leq l<u$. Suppose that we shift operation times on this interval such that all testing and processing times $t_j(x)$ and $p_j(x)$ are affine and non-decreasing functions in~$x$ with $t_j(0),p_j(0)\geq 0$, and the pairwise delays $d_x(i,j) + d_x(j,i)$ of $ALG$ for all pairs of jobs~$i\neq j$ satisfy the same properties, and further $OPT(l)>0$. Then, the worst-case ratio $\max_{x\in\left[l,u\right]} \nicefrac{ALG(x)}{OPT(x)}$ on this interval is attained for at least one $x^*\in\lbrace l, u\rbrace$.
\end{lemma}
\begin{proof}
     For each job $j$, let $q_j(x):= t_j(x) + p_j(x)$ for $x\in\left[l,u\right]$, which is non-decreasing and affine on this interval as $t_j, p_j$ are, and $q_j(0)\geq 0$. By Equation~\eqref{eq: Optimal cost}, the optimal cost can be written as $OPT(x)=\sum_j q_j(x) + \sum_{i<j} \min\lbrace q_i(x), q_j(x)\rbrace$. This yields a non-decreasing and piecewise-affine function. Affine functions are continuous and concave, and pointwise minima and sums preserve these properties. Hence, $OPT(x)$ is also a continuous and concave function.

    By Equation~\eqref{eq: Algorithmic cost}, the algorithmic cost equals $ALG(x)=\sum_j q_j(x) + \sum_{i<j}d_x(i,j) + d_x(j,i)$. Again, this is a non-decreasing, affine function, with $ALG(0)\geq 0$ being a sum of functions with these properties. An application of Lemma~\ref{lemma: maximum at the bounds} yields the claimed result.
\end{proof}
\begin{remark}\label{remark: technical aspects lemma 2}
    Lemma~\ref{lemma: shift to the bounds} will be applied to a class of index rules in Section~\ref{sec: Deterministic Algorithms} as follows. Consider index rules that assign testing operations priority values linear in $t_j$ and processing operations priority values linear in $p_j$, and always select the next operation to be scheduled from the set of available operations with the smallest priority value. Let $S$ be a fixed set of operations, all with the same priority value $y>0$. Suppose that we shift all testing or processing times in $S$ linearly such that their priority values are all equal to a value $x\in\left[l,u\right]$, with $0\leq l < y < u$, and keep all operations not in $S$ fixed and the same overall relative priority order of operations. Then, note that the pairwise delays $d_x(i,j) + d_x(j,i)$ are affine and non-decreasing in the shift $x$, and non-negative at $x=0$. Thus, we can apply Lemma~\ref{lemma: shift to the bounds} in this situation and iteratively reduce the number of different priority values and testing and processing times. 

    This result strictly generalizes the corresponding step in the instance-reduction framework of Dürr et al.~\cite{durr2020adversarial} for optional testing and uniform testing times, which restricts shifts to those of processing operations that preserve the order of operations in both the algorithm and the optimum. This suffices for uniform testing times, but our strict generalization is key for non-uniform testing times, see Theorem~\ref{thm: Instance Reduction}.

    The technical assumption $OPT(l)>0$ ensures that every occurring term is well-defined. This is equivalent to the existence of at least one strictly positive operation time. We assume this without loss of generality in the following. Lemma~\ref{lemma: shift to the bounds} also extends to randomized algorithms using expected pairwise delays. We will use this in Section~\ref{sec: Randomized Algorithms}.
\end{remark}

With Lemma~\ref{lemma: shift to the bounds}, we may iteratively find sequences of instances $I=I_0, I_1,I_2,\ldots,I_k$ such that $\nicefrac{ALG(I)}{OPT(I)}\leq \nicefrac{ALG(I_1)}{OPT(I_1)}\leq\ldots\leq \nicefrac{ALG(I_k)}{OPT(I_k)}$ for some $k$. 

The limitation of Lemma~\ref{lemma: shift to the bounds} is that it holds as long as the relative order of operations remains the same in the algorithm and the pairwise delays are affine. Since the algorithmic cost is not necessarily continuous under such shifts, when the order of operations changes, we will apply another \emph{adversarial} exchange argument as a second step in our instance-reduction framework. With this, we construct for an instance derived by the first step, $I_1\in\Psi_1\subsetneq \Psi$, another even more structured instance $I_2(\varepsilon)\in\Psi_2\subsetneq \Psi$ for every $\varepsilon>0$ such that $\nicefrac{ALG(I_2(\varepsilon))}{OPT(I_2(\varepsilon))}\geq\nicefrac{ALG(I_1)}{OPT(I_1)}-\varepsilon$.

Lemma~\ref{lemma: shift to the bounds} and the adversarial exchange argument will be used to identify a highly structured subclass of instances $\Psi_2\subsetneq \Psi$ for an algorithm $ALG$ such that 
\begin{equation}\label{eq: smaller instance class same ratio}
    \rho(ALG)=\sup_{I\in\Psi} \frac{ALG(I)}{OPT(I)}=\sup_{I\in\Psi_2} \frac{ALG(I)}{OPT(I)} \ ,
\end{equation} 
and for which an exact computation can easily be established, e.g., by solving a quadratic equation, see Section~\ref{sec: Exact Analysis Competitive Ratio}. 

\section{Deterministic Algorithms}\label{sec: Deterministic Algorithms}
In this section, we study the family of deterministic algorithms, $\beta$-SORT. Section~\ref{subsec: algorithmic family} describes this family, and Section~\ref{subsec: Limitations of Local Pairwise Bounds} shows limitations of previous approaches. In Section~\ref{sec: Exact Analysis Competitive Ratio}, we leverage our instance-reduction method from Section~\ref{sec: Instance-Reduction} to characterize the competitive ratio of $\beta$-SORT exactly and, in particular, show that $1$-SORT is $\varphi$-competitive for scheduling with obligatory testing, solving the open question on the exact competitive ratio by Dogeas et al.~\cite{dogeas2024obligatory}. 

\subsection{Algorithmic Family}\label{subsec: algorithmic family}
The family of deterministic algorithms $\beta$-\textsc{SORT}, as considered in~\cite{dogeas2024obligatory} for obligatory testing and similarly for optional testing, e.g. \cite{albers2021explorable,buld2025scheduling,durr2020adversarial,liu2023power}. It is a simple index rule as already described in Remark~\ref{remark: technical aspects lemma 2}: processing operations carry the priority value $p_j$, and testing operations carry the priority value $\beta t_j$, where $\beta>0$ is a single joint scaling factor that balances the priority between testing and processing.

\begin{definition}\label{def: Beta-SORT}
    The algorithm $\beta$-\textsc{SORT} is defined as follows: 
    Keep a priority queue and initially insert all testing operations with priority value $\beta t_j$. As long as the queue is not empty, schedule an operation with the smallest priority value and remove it from the queue. If it was a testing operation, add the processing operation with priority value $p_j$ to the queue.
\end{definition}

For $\beta=1$, we especially obtain a generalization of the \textsc{SPT}-rule. For a typical instance, $\beta$-\textsc{SORT} schedules more tests, i.e., it \emph{explores} more, in the earlier phase than the optimal algorithm. On structured subclasses of problems, e.g., for uniform testing times, a value of $\beta\neq 1$ may dominate $\beta=1$ in the worst-case, as we will see in Section~\ref{sec: Randomized Algorithms}.

Let $\rho(\beta)$ denote the competitive ratio of $\beta$-\textsc{SORT} for every $\beta>0$. Note that $\beta$-\textsc{SORT} is formulated as an index rule in Definition~\ref{def: Beta-SORT}, allowing for arbitrary tie-breaks between operations of the same priority value. However, the following result shows that the competitive ratio is independent of the tie-breaking rule for arbitrary testing times. Small perturbations to operation times within an instance can be used to enforce a specific order of operations without requiring tie-breaking.

\begin{lemma}\label{lemma: tie breaking}
    For an instance $I$ and any order of operations consistent with Definition~\ref{def: Beta-SORT}, we can find another instance $I_2$ with pairwise distinct priority values such that
    $\nicefrac{ALG(I_2)}{OPT(I_2)}\geq\nicefrac{ALG(I)}{OPT(I)}-\varepsilon$.
\end{lemma}
\begin{proof}
    Consider instance $I$, $\varepsilon>0$, and a given order of operations  $o_1,o_2,\ldots o_{2n}$ in $\beta$-SORT. Perturb the operation times using a small $\delta>0$ such that the priority values increase by $0,\delta,2\delta,\ldots,(2n-1)\delta$ in this order and that all priority values differ. This means that we either add $k\delta$ or $\nicefrac{k\delta}{\beta}$ to an operation time for some $k\in\lbrace0,1,\ldots 2n-1\rbrace$. This instance is denoted by $I_{\delta}$. Using Equations~\eqref{eq: Optimal cost} and~\eqref{eq: Algorithmic cost} for optimal and algorithmic cost similar as in the proof of Lemma~\ref{lemma: shift to the bounds}, we have that $ALG(I_{\delta})$ and $OPT(I_{\delta})$ are continuous in $\delta$ and can conclude that there is a $I_2:=I_{\delta'}$ for some $\delta'>0$ such that 
    $\nicefrac{ALG(I_2)}{OPT(I_2)}\geq \nicefrac{ALG(I)}{OPT(I)}-\varepsilon$.
\end{proof}

\subsection{Comparison of Analysis Methods}\label{subsec: Limitations of Local Pairwise Bounds}

An established way to prove upper bounds on the competitive ratio of algorithms for scheduling problems with testing, and in similar uncertainty frameworks, is a local pairwise delay analysis \cite{buld2025scheduling,dogeas2024obligatory,liu2023power,motwani1994nonclairvoyant}. Bounding $d(i,j) + d(j,i)\leq \theta\cdot\min\lbrace t_i+p_i, t_j+p_j\rbrace$ for all pairs of jobs $i\neq j$ yields an upper bound of $\theta$ using Equations~\eqref{eq: Optimal cost} and~\eqref{eq: Algorithmic cost}.

It is shown by Dogeas et al.~\cite{dogeas2024obligatory} that the best possible local pairwise bound for $1$-SORT equals $\theta=2$, considering an exhaustive analysis of possible orders. This analysis can be extended for general $\beta>0$; see Buld and Schulz~\cite{buld2025scheduling} and Liu, Liu, Wong, and Zhang~\cite{liu2023power} and their corresponding full versions, which study optional testing, with pairwise bounds for obligatory testing as a special case. The following proof is tailored for obligatory testing.

\begin{lemma}\label{lemma: Pairwise Bounds}
   For $\beta>0$, algorithm $\beta$-SORT is at most $\max\lbrace1+\beta,1+\nicefrac{1}{\beta}\rbrace$-competitive for scheduling with obligatory testing to minimize $\sum C_j$ on a single machine.
\end{lemma}
\begin{proof}
    Consider two jobs $i\neq j$ with $\min\lbrace t_i+p_i,t_j+p_j\rbrace=t_i+p_i$, i.e., job~$i$ is scheduled before job~$j$ in the optimal schedule. We bound the pairwise algorithmic delay $d(i,j)+d(j, i)$ based on a case distinction on the actions of $\beta$-SORT, similar to that in \cite{buld2025scheduling,dogeas2024obligatory,liu2023power}. Abbreviate by $\tau_j$ the testing operation and by $\pi_j$ the processing operation of job $j$ as in \cite{buld2025scheduling,dogeas2024obligatory}. We distinguish cases based on the order of operations in $\beta$-SORT.

\begin{enumerate}
    \item Order of operations $\left[\tau_i,\pi_i, \tau_j, \pi_j\right]$: $d(i,j)+d(j, i) = 1\cdot (t_i + p_i)$.
    \item $\left[\tau_i, \tau_j, \pi_i, \pi_j\right]$: $d(i,j)+d(j, i) = t_i + p_i + t_j \leq t_i + (1+\nicefrac{1}{\beta})p_i\leq(1+\nicefrac{1}{\beta})( t_i + p_i)$, tight for $t_i=0$, $t_j=\nicefrac{p_i}{\beta}$, and $p_i\leq p_j$.
    \item $\left[\tau_i, \tau_j, \pi_j, \pi_i\right]$: $d(i,j)+d(j, i) = t_i  + t_j + p_j \leq t_i + (1+\nicefrac{1}{\beta})p_i\leq(1+\nicefrac{1}{\beta})( t_i + p_i)$, tight for $t_i=0$, $t_j=\nicefrac{p_i}{\beta}$, $p_j=p_i$.
    \item $\left[\tau_j, \tau_i, \pi_i, \pi_j\right]$: $d(i,j)+d(j, i) = t_j + t_i + p_i \leq 2t_i + p_i\leq2(t_i + p_i)$, tight for $t_i=t_j$, $p_i=0$.
    \item $\left[\tau_j, \tau_i, \pi_j, \pi_i\right]$: $d(i,j)+d(j, i) = t_j + t_i + p_j \leq 2t_i + p_i\leq2(t_i + p_i)$, tight for $t_i=t_j$, $p_i=p_j=0$.
    \item $\left[\tau_j, \pi_j, \tau_i, \pi_i\right]$: $d(i,j)+d(j, i) = t_j + p_j \leq (1+\beta)t_i \leq (1+\beta)(t_i + p_i)$, tight for $t_i=t_j$ $t_i = \nicefrac{p_j}{\beta}$, $p_i=0$.
\end{enumerate}
\end{proof}

Using this local pairwise analysis, no better bound than $2$ can be shown for any $\beta>0$. To overcome this limitation, Dogeas et al.~\cite{dogeas2024obligatory} give a method to amortize the delay of pairs with local delay close to~$2$ for $1$-SORT and achieve an upper bound of $1.861$ on its competitive ratio. 
They also give a family of instances to establish a lower bound on the competitive ratio for the algorithmic family $\beta$-SORT. We will show in Section~\ref{sec: Exact Analysis Competitive Ratio} how to match this lower bound using our instance-reduction method from Section~\ref{sec: Instance-Reduction}. The tight pairs of jobs in the proof of Lemma~\ref{lemma: Pairwise Bounds} will be crucial in this.

\begin{theorem}\label{thm: Exact Competitive Ratio}
    The competitive ratio of $\beta$-SORT is exactly
    \begin{equation}
       \rho(\beta) =\left\lbrace\begin{array}{ll}
             \frac{1}{2}\left(1+\sqrt{\frac{\beta + 4}{\beta}}\right),&  0 < \beta \leq 1, \\
             \frac{1}{2\beta}\left(1+\sqrt{4\beta\left(\beta^2+\beta -1\right)+1}\right),& \beta \geq 1,
        \end{array}\right.
    \end{equation}
    and this expression is minimized at $\beta=1$ attaining a value of $\rho(1)=\varphi$.
\end{theorem}
The proof of Theorem~\ref{thm: Exact Competitive Ratio} is relegated to Section~\ref{sec: Exact Analysis Competitive Ratio}.
\begin{remark}
    In Figure~\ref{fig:Bounds on the Competitive Ratio}, we depict various bounds on the competitive ratio as $\beta>0$ varies. These include the local pairwise delay bound from Lemma~\ref{lemma: Pairwise Bounds}, the previous best bound by Dogeas et al.~\cite{dogeas2024obligatory} for $1$-SORT, and the exact competitive ratios obtained in this paper. Interestingly, both the local pairwise delay and the exact competitive ratio attain their minima at $\beta=1$, with values of $2$ and $\varphi$, respectively. Furthermore, this shows that a local pairwise delay analysis is not sufficient for this algorithmic family.
\end{remark}

\begin{figure}[htb]
    \centering
    \includegraphics[width=0.66\linewidth]{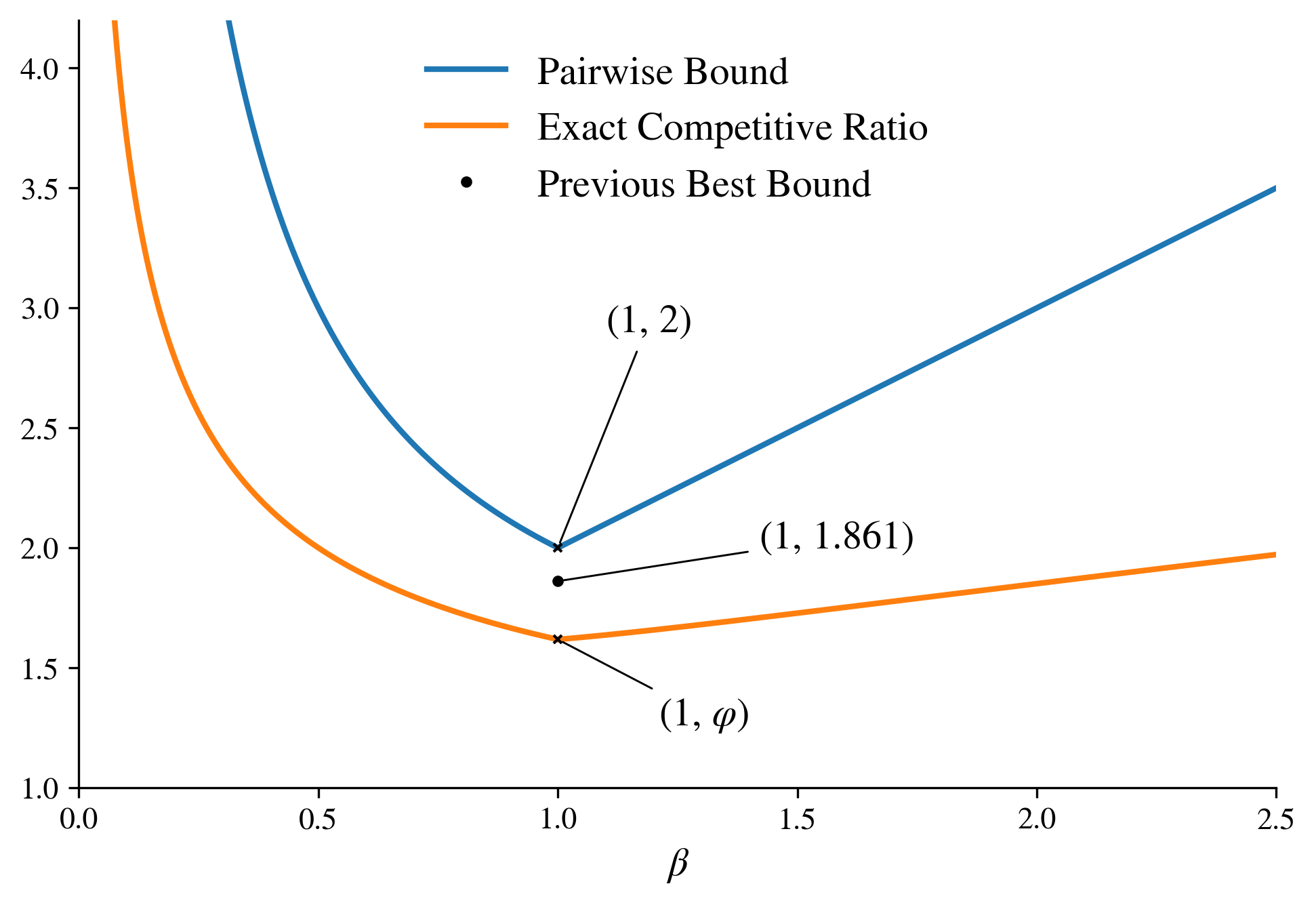}
    \caption{Bounds on the competitive ratio using previous approaches and our instance-reduction method.}
    \label{fig:Bounds on the Competitive Ratio}
\end{figure}

 Note that our exact competitive ratio of $\varphi$ for $1$-SORT for arbitrary testing times almost matches the lower bound of about $1.57$ \cite{rozhon2026test} for any deterministic algorithm and uniform testing times.

\subsection{Exact Analysis of the Competitive Ratio of $\beta$-SORT}\label{sec: Exact Analysis Competitive Ratio}
Using our instance-reduction method from Section~\ref{sec: Instance-Reduction}, we can show that there are worst-case instances with only a few distinct parameters in the operation times. First, we iteratively merge groups of operations with the same priority value. Then, we apply adversarial exchange arguments to obtain worst-case instances for which we compute the competitive ratio exactly. 

\begin{theorem}\label{thm: Instance Reduction}
    For each instance $I$ of the obligatory testing problem, and any order of operations consistent with Definition~\ref{def: Beta-SORT}, there is an instance~$I'$ such that $\beta$-SORT yields the same order of operations and $\nicefrac{ALG(I)}{OPT(I)}\leq \nicefrac{ALG(I')}{OPT(I')}$, where each job $j$ is of the structure as follows: 
\end{theorem}
    \begin{table}[htb]
        \centering
        \caption{Worst-case structure for $\beta$-SORT.}
        \label{tab:worst-case example}
        \begin{tabular}{c|c|c|c}
        Job Parameters & Type $1$ & Type $2$ & Type $3$ \\
        \hline
        $(t_j,p_j)$ & $(0,\beta)$ & $(1,\beta)$ & $(1,0)$ 
        \end{tabular}
    \end{table}
\begin{proof}
Let $I$ be an arbitrary instance, and $X(I):=(\lbrace \beta t_1, \beta t_2,\ldots, \beta t_n\rbrace\cup\lbrace p_1,p_2,\ldots,p_n\rbrace)\setminus\lbrace 0 \rbrace$ denote the set of strictly positive priority values in $\beta$-SORT for all operations (for $1$-SORT, these are exactly the operation times). 

We rescale all operation times by the maximal priority value $x_{\max}:=\max_{x_i\in X(I)} x_i$, such that we have $x_{\max}=\beta$ and $x_i\in\left[ 0,\beta\right]$. This is without loss of generality, since the order of priority values and the ratio remain unchanged.

If $\vert X(I)\vert = 1$, we are done. Otherwise, we will use our instance-reduction method to construct a series of instances $I_1, I_2, \ldots, I_k$ such that $\vert X(I)\vert \geq \vert X(I_1)\vert \ldots \geq \vert X(I_k)\vert =1$ and $\nicefrac{ALG(I)}{OPT(I)}\leq \nicefrac{ALG(I_1)}{OPT(I_1)}\leq \ldots \leq \nicefrac{ALG(I_k)}{OPT(I_k)}$ for some $k\leq \vert X(I)\vert -1 \leq 2n -1$ and the result follows with $I':= I_k$.

We use Lemma~\ref{lemma: shift to the bounds} via Remark~\ref{remark: technical aspects lemma 2}. In particular, if we select a subset of operations and shift their priority values in $\beta$-SORT uniformly while preserving the order of operations, especially for tie-breaks, the pairwise delay $d(i,j) + d(j,i)$ between two jobs $i\neq j$ is affine in the shift.

First, note that for each job $j$ with $0<p_j<\beta t_j$, $\beta$-SORT schedules the processing operation of job $j$ immediately after $j$'s test. Hence, by applying Lemma~\ref{lemma: shift to the bounds} independently to each such job $j$ with interval $\left[ 0, \beta t_j \right]$, we can conclude that we can restrict ourselves to $p_j\in\lbrace 0,\beta t_j\rbrace$ for these jobs without changing the algorithmic order and without worsening the competitive ratio. Let the resulting instance be denoted by $I_1$.

Second, for the instance $I_1$, we write $0 =: x_0 < x_1 < x_2 < \ldots < x_{\vert X(I_1)\vert}$ for all priority values in $x_i\in X(I_1)\cup\lbrace 0 \rbrace$ that $\beta$-SORT uses on this instance. Pick an arbitrary index $i\in\lbrace 1, \ldots, \vert X(I_1)\vert -1\rbrace$ together with all operations of priority values $x_i$. Apply Lemma~\ref{lemma: shift to the bounds} to all these operations and the interval $\left[ x_{i-1}, x_{i+1}\right]$. Note that this shift in operation times does not change the relative order of operations since $\beta$-SORT always uses the smallest priority value, and if there is a job with $p_j=0$, it will still be scheduled immediately after the test of job $j$. Hence, by Lemma~\ref{lemma: shift to the bounds}, we obtain $I_2$ from $I_1$ with $ \vert X(I_1)\vert  > \vert X(I_2)\vert$ and $\nicefrac{ALG(I_1)}{OPT(I_1)}\leq \nicefrac{ALG(I_2)}{OPT(I_2)}$ by shifting all operations of priority value $x_i$ to a priority value of either $x_{i-1}$ or $x_{i+1}$. 

Iteratively following this approach, we arrive at an instance $I'$ with $\nicefrac{ALG(I)}{OPT(I)}\leq \nicefrac{ALG(I')}{OPT(I')}$ and $\vert X(I')\vert = 1$. Hence, instance $I'$ only has priority values $0$ and $\beta$, implying that for all jobs $j$, we have $t_j\in\lbrace 0, 1\rbrace$ and $p_j\in\lbrace 0, \beta\rbrace$. Removing jobs with $(t_j,p_j)=(0,0)$ that neither contribute to the algorithmic nor the optimal cost, leaves us with the pairs of operation times $(t_j,p_j)\in\lbrace (0, \beta), (1,\beta), (1,0) \rbrace$ as stated.
\end{proof}

Theorem~\ref{thm: Instance Reduction} shows that we can restrict our attention to instances of the three types above. These are exactly the job pairs where at least one pairwise delay bound is tight; see Lemma~\ref{lemma: Pairwise Bounds}. Since this step does not change the algorithmic order, we must consider adversarial tie-breaking when analyzing the competitive ratio for this class of instances. This is established by Lemma~\ref{lemma: Adversarial Exchange Lemma} using adversarial exchange arguments.

\begin{lemma}\label{lemma: Adversarial Exchange Lemma}
    The worst-case schedule of $\beta$-SORT on instances in Table~\ref{tab:worst-case example} is structured as follows: First jobs of type $(0,\beta)$ are tested instantly, then jobs of type $(1,\beta)$ are tested, then,
    \begin{enumerate}
        \item if $\beta\geq 1$, all of these jobs are completed in arbitrary order, and then all jobs of type $(1,0)$ are tested and immediately completed.
        \item if $\beta\leq 1$, all jobs of type $(1,0)$ are tested and immediately completed, and then all jobs of type $(0,\beta)$ and $(1,\beta)$ are completed in arbitrary order.
    \end{enumerate}
\end{lemma}
\begin{proof}
    The proof is a straightforward exchange argument; using that the optimal cost is constant, and the pairwise delay $d(i,j)+d(j,i)$ is maximized for any two jobs $i\neq j$ of these three types when operations are scheduled exactly in the stated order for the two cases. Hence, all pairwise delays are maximized simultaneously in the described schedules.
\end{proof}

We first use Theorem~\ref{thm: Instance Reduction} and Lemma~\ref{lemma: Adversarial Exchange Lemma} to prove our main result for $\beta=1$ and then extend the proof strategy to the other cases of $\beta\neq 1$ afterwards.

\begin{theorem}\label{thm: 1-SORT}
    $1$-SORT is $\varphi$-competitive for scheduling with obligatory testing on a single machine to minimize $\sum C_j$.
\end{theorem}
\begin{proof}
    Following Theorem~\ref{thm: Instance Reduction}, we can focus on instances with three types of jobs, with parameters $(0,1),(1,1),(1,0)$ in the case of $\beta=1$. By Lemma~\ref{lemma: Adversarial Exchange Lemma}, Case 1, $1$-SORT schedules all tests of jobs of type $(0,1)$ immediately, followed by all tests of jobs of type $(1,1)$, and then completes all jobs in arbitrary order, taking time $1$ each. In the optimum, the jobs of types $(0,1)$ and $(1,0)$ are completed in arbitrary order first, then all jobs of type $(1,1)$ are completed without preemption between testing and processing. In Table~\ref{tab: beta = 1}, we explicitly write down the pairwise algorithmic and optimal delays $d(i,j)+d(j,i) \mid d^*(i,j)+d^*(j,i)$ for jobs $i\neq j$ of the two types.
    
    \begin{table}[htb]
        \centering
       \caption{Pairwise delays of $1$-SORT and the optimal schedule for the worst-case instances.}       
       \label{tab: beta = 1}
        \begin{tabular}{cc|r@{\;$\mid$\;}l r@{\;$\mid$\;}l r@{\;$\mid$\;}l}
        && \multicolumn{2}{c}{$(0,1)$} & \multicolumn{2}{c}{$(1,1)$} & \multicolumn{2}{c}{$(1,0)$} \\
        \hline
        Type $1$ &$(0,1)$ & $1$ & $1$ & $2$ & $1$ & $1$ & $1$ \\
        Type $2$ &$(1,1)$ & \multicolumn{2}{c}{} & $3$ & $2$ & $2$ & $1$ \\
        Type $3$ &$(1,0)$ & \multicolumn{2}{c}{} & \multicolumn{2}{c}{} & $1$ & $1$
        \end{tabular}
    \end{table}
    
    Consider an arbitrary instance $I$ consisting of $\mu n, \lambda n,\sigma n$ jobs of types $(0,1),(1,1),(1,0)$, where $\mu,\lambda,\sigma\geq 0$ and $\mu+\sigma=1-\lambda$. Hence, using Equation~\eqref{eq: Algorithmic cost} and Table~\ref{tab: beta = 1}:  
    \begin{equation}
        \begin{aligned}
            ALG &=  n\left((1-\lambda) + 2\lambda\right) + 2\lambda(1-\lambda)n^2 + 3\frac{\lambda n(\lambda n -1)}{2} + \frac{(1-\lambda)n((1-\lambda)n -1)}{2}\\
            &=  \frac{n^2}{2} \left( 2\lambda +1 \right) + \frac{n}{2} \ ,
        \end{aligned}
    \end{equation}
    and using Equation~\eqref{eq: Optimal cost} and Table~\ref{tab: beta = 1}:
        \begin{equation}
        \begin{aligned}
         OPT &=  n\left((1-\lambda) + 2\lambda\right) + \lambda(1-\lambda)n^2 + 2\frac{\lambda n(\lambda n -1)}{2} + \frac{(1-\lambda)n((1-\lambda)n -1)}{2}\\
         &=  \frac{n^2}{2} \left( \lambda^2 +1 \right) + \frac{n}{2}\left( \lambda  + 1 \right) \ .
        \end{aligned}
    \end{equation}
    We can bound
    \begin{equation}\label{eq: ratio computation beta = 1}
        \frac{ALG}{OPT} \leq \frac{2\lambda +1}{\lambda^2 +1} 
    \end{equation}
    for all such instances. For each fixed $\lambda$, the competitive ratio of $1$-SORT restricted to the class of instances with this parameter is bounded by the right-hand side. Moreover, for fixed $\lambda$, the ratio approaches the latter ratio as $n\rightarrow\infty$, so we can focus on the coefficients of the quadratic terms. The derivative is equal to $\nicefrac{2(-\lambda^2-\lambda +1)}{(\lambda^2 +1)^2}$, and the positive root of the latter quadratic numerator set to $0$ is equal to $\lambda^*=\nicefrac{(\sqrt{5}-1)}{2}=\varphi-1=\nicefrac{1}{\varphi}\approx 0.618$. Hence, we can calculate, using $\varphi^2-\varphi-1=0$ twice,
    \begin{equation}\label{eq:Golden Ratio Calculation}
        \rho(1)-\varphi = \frac{2(\varphi -1) +1}{ (\varphi -1)^2 +1} - \varphi = \frac{2\varphi -1}{\varphi^2 -2\varphi+2}-\varphi = \frac{2\varphi -1}{3-\varphi}-\varphi = \frac{\varphi^2-\varphi -1}{3-\varphi} = 0 \ ,
    \end{equation}
    and obtain the competitive ratio of $\rho(1)=\varphi$.
\end{proof}

We have proven that $1$-SORT is $\varphi$-competitive. Note that, for the identified instances, some pairs still have pairwise delays that are off the optimal pairwise delays for such jobs by a factor of $2$; see Table~\ref{tab: beta = 1}. However, there are pairs of ratios $3/2$ and $1$, as well, which offset the overall ratio again. Following the proof of Theorem~\ref{thm: 1-SORT}, the limiting worst-case instance has a fraction of $\lambda^* = \nicefrac{1}{\varphi}$ jobs of type~$(1,1)$ and a fraction of $1-\lambda^*$ jobs that is an arbitrary combination of the other two types.

For the cases of $\beta\neq 1$, we can extend the proof strategy of Theorem~\ref{thm: 1-SORT} to further reduce the instances to be considered and then also perform a single-parameter optimization.

\begin{theorem}\label{thm: worst-case families}
Consider $\beta$-SORT with an adversarial tie-breaking rule as in Lemma~\ref{lemma: Adversarial Exchange Lemma}.

If $\beta \geq 1$, then $\beta$-SORT has the same competitive ratio when the instance class is restricted to jobs of uniform testing time and two distinct processing times, i.e., for $\Psi_\beta := \lbrace I \vert (t_j,p_j)\in\lbrace(1,\beta),(1,0)\rbrace \textup{ for all } j\rbrace$, we have $\rho(\beta)=\sup_{I\in \Psi_\beta} \nicefrac{ALG(I)}{OPT(I)}$.

If $\beta \leq 1$, then $\beta$-SORT has the same competitive ratio when the instance class is restricted to jobs of uniform processing time and two distinct testing times, i.e., for $\Psi_\beta:= \lbrace I \vert (t_j,p_j)\in\lbrace(0,\beta),(1,\beta)\rbrace \textup{ for all } j\rbrace$, we have $\rho(\beta)=\sup_{I\in \Psi_\beta} \nicefrac{ALG(I)}{OPT(I)}$.
\end{theorem}
\begin{proof}[Proof Idea]
    The proof is similar to that for $1$-SORT in Theorem~\ref{thm: 1-SORT}, with more variables and a few more small arguments.
    Following Theorem~\ref{thm: Instance Reduction}, we consider an arbitrary instance~$I$ consisting of jobs of types $(0,\beta)$, $(1,\beta)$, $(1,0)$. We compute the pairwise delays using adversarial tie-breaks and obtain a similar structure to that in Table~\ref{tab: beta = 1}:

    \begin{table}[htb]
    \begin{minipage}[h]{0.49\textwidth}
        \centering
       \caption{Pairwise delays of $\beta$-SORT and the optimal schedule for the worst-case instances and $\beta>1$.}
       \label{tab: beta > 1}
        \begin{tabular}{c|r@{\;$\mid$\;}l r@{\;$\mid$\;}l r@{\;$\mid$\;}l}
        & \multicolumn{2}{c}{$(0,\beta)$} & \multicolumn{2}{c}{$(1,\beta)$} & \multicolumn{2}{c}{$\hspace{4.2ex}(1,0)$} \\
        \hline
        $(0,\beta)$ & $\beta$ & $\beta$ & $1+\beta$ & $\beta$ & $\beta$ & $1$ \\
        $(1,\beta)$ & \multicolumn{2}{c}{} & $2+\beta$ & $1+\beta$ & $1+\beta$ & $1$ \\
         $(1,0)$ & \multicolumn{2}{c}{} & \multicolumn{2}{c}{} & $1$ & $1$
        \end{tabular}
    \end{minipage}\hfill
    \begin{minipage}[h]{0.49\textwidth}
            \centering
    \caption{Pairwise delays of $\beta$-SORT and the optimal schedule for the worst-case instances and $\beta<1$.}
        \label{tab: beta < 1}
        \begin{tabular}{c|r@{\;$\mid$\;}l r@{\;$\mid$\;}l r@{\;$\mid$\;}l}
        & \multicolumn{2}{c}{$(0,\beta)$} & \multicolumn{2}{c}{$(1,\beta)$} & \multicolumn{2}{c}{$(1,0)$} \\
        \hline
        $(0,\beta)$ & $\beta$ & $\beta$ & $1+\beta$ & $\beta$ & $1$ & $\beta$ \\
        $(1,\beta)$ & \multicolumn{2}{c}{} & $2+\beta$ & $1+\beta$ & $2$ & $1$ \\
        $(1,0)$ & \multicolumn{2}{c}{} & \multicolumn{2}{c}{} & $1$ & $1$
        \end{tabular}
    \end{minipage}
    \end{table}
    
    We compute the algorithmic and optimal costs for each. In the case of $\beta>1$, we argue that we can convert all jobs of type~$(0,\beta)$ to either being jobs of type $(1,\beta)$ or jobs of type~$(1,0)$, reaching $\rho(\beta)$ as a limiting ratio, as well. Vice versa, for $\beta< 1$, we can convert jobs of type $(1,0)$ similarly. The detailed proof steps and calculations are relegated to Appendix~\ref{appendix: calculation beta reduction}.
\end{proof}

\begin{remark}\label{remark: perturbations worst-case}
    Note that in the instances in Theorem~\ref{thm: worst-case families}, adversarial tie-breaking is crucial. However, we can derive tight instances that do not rely on adversarial tie-breaking but yet enforce the specified order, following Lemma~\ref{lemma: tie breaking}. Another perturbation for some small $\varepsilon>0$ to achieve this is the following: for $\beta \leq 1$, use jobs of type $(0,\beta+\varepsilon), (1,\beta+\varepsilon)$, and for $\beta \geq 1$, use jobs of type $(1,\beta+\varepsilon), (1+2\epsilon,0)$ similar as in the class of instances in \cite{dogeas2024obligatory}.
\end{remark}

Having established Theorem~\ref{thm: worst-case families}, we can compute the ratio exactly for these highly structured instances for $\beta\neq 1$. This is the last step of the proposed instance-reduction method from Section~\ref{sec: Instance-Reduction}. For instances of this structure (up to scaling and perturbations), this calculation was already established in~\cite{dogeas2024obligatory} to obtain a lower bound on the competitive ratio~$\rho(\beta)$. We briefly repeat the results here to conclude the proof of Theorem~\ref{thm: Exact Competitive Ratio}.

\begin{proof}[Proof of Theorem~\ref{thm: Exact Competitive Ratio}]
    Consider an instance $I$ consisting of $\mu n, \lambda n,\sigma n$ jobs of types $(0,\beta)$, $(1,\beta)$, $(1,0)$, where $\mu,\lambda,\sigma\geq 0$ and $\mu+\lambda+\sigma=1$ as in Theorem~\ref{thm: Instance Reduction}. Using Theorem~\ref{thm: worst-case families}, the formulas for $\nicefrac{ALG}{OPT}$ from Appendix~\ref{appendix: calculation beta reduction} and substitution of variables, we need to solve a maximization problem in a single parameter $\lambda\in[0,1]$ depending on the value of~$\beta$ for the two cases $\beta>1$ and $\beta<1$. 
    
    For $\beta>1$, one obtains the fraction of 
    $\lambda^*=\nicefrac{(-2\beta + 1 + \sqrt{(2\beta-1)^2+4\beta^3})}{2\beta^2}$ and $\sigma^*=1-\lambda^*$, $\mu^*=0$, in the worst-case family, see also \cite{dogeas2024obligatory}, resulting in the competitive ratio as stated.
    
    For $\beta<1$, one obtains the fraction of 
    $\lambda^*=\nicefrac{(\sqrt{\beta^2+4\beta}-\beta)}{2}$ and $\mu^*=1-\lambda^*$, $\sigma^*=0$, in the worst-case family, see also \cite{dogeas2024obligatory}, resulting in the competitive ratio as stated.
\end{proof}

We end this section with a remark on the special case of uniform testing times, i.e., $t_j=1$ for each job $j$, to distinguish it from the general case.

\begin{remark}\label{remark: Uniform Testing Times}
    On the class of instances with uniform testing times, the competitive ratio of $\beta$-SORT does not change for $\beta\geq 1$ in comparison to the case of arbitrary testing times, following Theorem~\ref{thm: worst-case families}. For $\beta< 1$, following Theorem~\ref{thm: Instance Reduction}, one obtains worst-case instances with $(t_j,p_j)\in\lbrace(1,\beta),(1,0)\rbrace$ for all $j$. The competitive ratio is enhanced compared to Theorem~\ref{thm: Exact Competitive Ratio} on this branch. But the resulting expression is also minimized for $\beta=1$, still yielding the Golden Ratio as the competitive ratio of $1$-SORT on this class of instances. 

    Unlike in the general case, tie-breaking can make a difference for instances with uniform testing times $t_j=1$, and thereby improve competitive ratios. The algorithm ($\beta$-)\textsc{SIDLE}~\cite{dogeas2024obligatory} equals $\beta$-SORT with an additional tie-break in favor of processing operations whenever there is a processing operation of priority value $p_j = \beta t_j =\beta$. For $\beta> 1$ and the instance structure from Theorem~\ref{thm: worst-case families} with jobs of types $(1,\beta),(1,0)$, this tie-breaking prevents a job of type~$(1,\beta)$ from achieving the maximally possible pairwise delay with other jobs of type $(1,\beta)$ and of type $(1,0)$ at once, see Table~\ref{tab: beta > 1}. Note that this effect vanishes again for $\beta=1$, considering the perturbed instance $(1,1+\varepsilon),(1,0)$ for small $\varepsilon>0$.
    
    Using our instance-reduction method of Section~\ref{sec: Instance-Reduction}, tailored for instances with uniform testing times, or that of Dürr et al.~\cite{durr2020adversarial}, one can show as in \cite{dogeas2024obligatory} that we can restrict our analysis to the worst-case structure that uses $p_j\in\lbrace 0, \beta, \beta + \varepsilon\rbrace$ for some arbitrarily small $\varepsilon>0$ and an adversarial order of testing the jobs (in non-increasing order of processing times). Explicitly calculating the competitive ratio in this case, Dogeas et al.~\cite{dogeas2024obligatory} show that the parameter $\beta^*\approx 1.355$ yields the minimal competitive ratio of $\rho(\beta^*)\approx 1.585$ for $\beta$-\textsc{SIDLE}. This also shows that it might be beneficial to consider tie-breaking rules and a value $\beta\neq 1$ for specially structured instances, in contrast to the setting with arbitrary testing time, see Lemma~\ref{lemma: tie breaking}.

    Hence, we conclude that the competitive ratio obtained for $\beta$-SORT crucially depends on the instance having distinct testing times that an adversary can structurally exploit (see Remark~\ref{remark: perturbations worst-case}) or on adversarial tie-breaking respectively.
\end{remark}

\section{Randomized Algorithms for Uniform Testing Times }\label{sec: Randomized Algorithms}

In this section, we consider a randomized variant of $\beta$-\textsc{SORT} on instances with uniform test times, in which we uniformly randomize the test order and then run $\beta$-\textsc{SORT} for a fixed value of $\beta>0$. As a tie-breaker, we may also prioritize processing over testing for all pairs of all jobs as for $\beta$-SIDLE \cite{dogeas2024obligatory}. This randomization step accounts for the fact that $\beta$-\textsc{SORT} tests jobs in the worst-case instances in exactly the opposite order to the optimal algorithm. This algorithm is a variant of the \textsc{RANDOM} algorithm for optional testing of Dürr et al.~\cite{durr2020adversarial} in the obligatory testing setting. 

In Section~\ref{sec: Exact Competitive Ratio Random}, we will compute the exact competitive ratio for this rule and show that it is minimized for $\beta^*\approx 1.947$ with a competitive ratio of $\rho(\beta^*)\approx 1.359$. In Section~\ref{sec: Lower Bounds for Randomized Algorithms}, we use a connection to the literature on stochastic scheduling with testing \cite{levi2024scheduling} via Yao's principle~\cite{yao1977probabilistic} that has not been used in the literature on adversarial models for scheduling with testing up to now giving us a lower bound of $\nicefrac{4}{3}$ for any randomized algorithm. This bound has also been derived differently in \cite{rozhon2026test} with a matching upper bound for an intricate sampling algorithm. We give a novel structural interpretation here and further show that a simple index rule, paired with randomization, is almost best possible for uniform testing times, as well. This section also shows that a significant part of the competitive ratio of $\beta$-SORT is due to the adversarial testing order.

\subsection{Exact Competitive Ratio}\label{sec: Exact Competitive Ratio Random}
The expected pairwise delay between jobs $i\neq j$ with $p_i\leq p_j$ by $\beta$-\textsc{SIDLE} for uniform testing times, i.e., $t_j=1$ for all jobs $j$, for uniformly permuted tests is the following:
\begin{equation}\label{eq: pairwise random ordering}
    \mathbb{E}\left[d(i,j) + d(j,i)\right]
    = 1 + \left\lbrace\begin{array}{cl}
     \frac{p_i + p_j}{2}& , 0\leq p_i,p_j\leq \beta, \\
     \frac{1}{2} + p_i  & , 0\leq p_i \leq \beta < p_j, \\
     1 + p_i & , \beta < p_i,p_j.
\end{array}\right.
\end{equation}
For the pairwise delay in the optimal solution, we have, as before, $d^*(i,j) + d^*(j,i) = 1+ \min\lbrace p_i,p_j\rbrace = 1+ p_i$. As in Section~\ref{sec: Deterministic Algorithms}, we start by computing pairwise bounds.

\begin{lemma}\label{lemma: Expected pairwise bounds}
    For uniform testing times, $\beta$-\textsc{SIDLE} on a uniformly permuted test order achieves a competitive ratio of at most
   $\max\left\lbrace 1+ \nicefrac{\beta}{2},1+\nicefrac{1}{1+\beta}\right\rbrace$. This expression is minimized for $\beta=1$ with $\rho(1)\leq \nicefrac{3}{2}$.
\end{lemma}
\begin{proof}
    The bound on the ratio is achieved by a case distinction on Equation~\eqref{eq: pairwise random ordering}. We obtain a factor of at most $1+\nicefrac{(p_j-p_i)}{2(1+p_i)}\leq 1+ \nicefrac{\beta}{2}$ in the first case, tight for $p_i=0, p_j=\beta$; of $\nicefrac{(\nicefrac{3}{2}+p_i)}{(1+p_i)}\leq\nicefrac{3}{2}$ in the second case, tight for $p_i=0$; and of $\nicefrac{(2+p_i)}{(1+p_i)}\leq 1+\nicefrac{1}{(1+\beta)}$ in the last case, tight for $p_i=\beta$. The bound of $\nicefrac{3}{2}$ in the second case is dominated by the maximum of the bounds in the other two cases.
\end{proof}

As in Section~\ref{sec: Deterministic Algorithms}, we can leverage our instance-reduction scheme from Section~\ref{sec: Instance-Reduction} or that of Dürr et al.~\cite{durr2020adversarial} to consider all pairs in a more amortized fashion and characterize the exact competitive ratio, similar to that of the algorithms RANDOM in \cite{durr2020adversarial} and $\beta$-SIDLE in \cite{dogeas2024obligatory}.

\begin{theorem}\label{thm: Optimized Bounds}
    For uniform testing times, $\beta$-\textsc{SIDLE} on a uniformly permuted test order achieves a competitive ratio of exactly
    \begin{equation}\label{eq: randomized SIDLE exact competitive ratio}
       \rho(\beta) =\left\lbrace\begin{array}{ll}
            1+\frac{1}{1+\beta},&  0 < \beta\leq 1, \\
                \max\left\lbrace1+\frac{1}{2\sqrt{\beta}},\frac{1+\sqrt{1+\beta}}{2} \right\rbrace,& \beta \geq 1.
        \end{array}\right.
    \end{equation}
    This expression is minimized for $\beta^*\approx 1.947$ being the solution of $\beta^2=2\sqrt{\beta}+1$ and attaining a value of $\rho(\beta^*)\approx 1.359$.
\end{theorem}
\begin{proof}
Let $p_1\leq p_2\leq \ldots \leq p_n$. Based on the Equation~\eqref{eq: pairwise random ordering} for the pairwise delay, we can first apply the instance-reduction method and then explicitly optimize over the parameter set to obtain a closed-form solution. Since $\mathbb{E}\left[d(i,j) + d(j,i)\right]$ is linear when shifting processing times in $[0,\beta]$, we can apply Lemma~\ref{lemma: shift to the bounds} to restrict all processing times in this interval to be either $0$ or $\beta$. For the competitive ratio, we obtain:
\begin{equation}\label{eq: Surplus Rewrite}
    \frac{ALG}{OPT} = \frac{OPT + \sum_{i<j} \eta_{ij} }{OPT} = 1 + \frac{\sum_{i<j} \eta_{ij} }{\sum_j (n+1-j)(1+p_j)} \ ,
\end{equation}
with short-hand notation $\eta_{ij}$ equal to $\nicefrac{(p_j-p_i)}{2}$ for $0\leq p_i,p_j\leq \beta$, $\nicefrac{1}{2}$ for $0\leq p_i \leq \beta < p_j$, and $1$ for $\beta < p_i,p_j$.
All processing times with $p_j>\beta$ can be shifted to $\beta+\varepsilon$ for small $\varepsilon>0$, since they only appear in the denominator and with a positive sign.
Hence, we can restrict ourselves to instances with three different processing times $p_j\in\lbrace0,\beta,\beta+\varepsilon\rbrace$. Similarly as in \cite{dogeas2024obligatory} and Section~\ref{sec: Deterministic Algorithms}, for arbitrary fractions $\mu,\lambda \in [0,1]$, let there be $\mu\lambda n$ jobs of length $0$, $\mu(1-\lambda) n$ jobs of length $\beta$, and $
(1-\mu) n$ jobs of length $\beta+\varepsilon$. For this family, we obtain:
\begin{equation}
    \begin{aligned}
    \sum_{i<j} \eta_{ij} &= \frac{\beta}{2}\cdot \mu^2\lambda(1-\lambda)n^2 + \frac{1}{2}\cdot(1-\mu)\mu n^2 +  1\cdot\frac{1}{2}(1-\mu)n((1-\mu)n -1)\\ &= \frac{n^2}{2}\left( 1-\mu +\beta\mu^2\lambda(1-\lambda)\right) - \frac{n}{2}(1-\mu) \ .
    \end{aligned}
\end{equation}

For the optimal schedule for this family, we obtain, see \cite{dogeas2024obligatory}:

\begin{equation}
    OPT=\frac{n^2}{2}\cdot (1+\beta(\mu\lambda -1)^2) + \frac{n}{2}\cdot \left(\mu + \beta\mu(1-\lambda) + (1+\beta+\varepsilon)(1-\mu)\right) \ .
\end{equation}

We can bound
\begin{equation}\label{eq: Bound on the Ratio Randomized}
    \frac{ALG}{OPT} \leq 1 + \max_{0\leq\mu,\lambda\leq1}\frac{1-\mu +\beta\mu^2\lambda(1-\lambda)}{1+\beta(\mu\lambda -1)^2}\ ,
\end{equation}
and obtain a maximization problem in two variables. Solving this for each $\beta$ and minimizing the expression yields $\beta^*$ as stated. The detailed algebraic steps are relegated to Appendix~\ref{appendix: parameter optimization}.
\end{proof}

We conclude this section with a comparison of the exact competitive ratios of $\beta$-\textsc{SORT} for arbitrary testing times, $\beta$-\textsc{SIDLE} for uniform testing times, and $\beta$-\textsc{SIDLE} for uniformly ordered tests.
 
\begin{remark}\label{remark: comparison of bounds}
\begin{figure}[htb]
    \centering
    \includegraphics[width=0.66\linewidth]{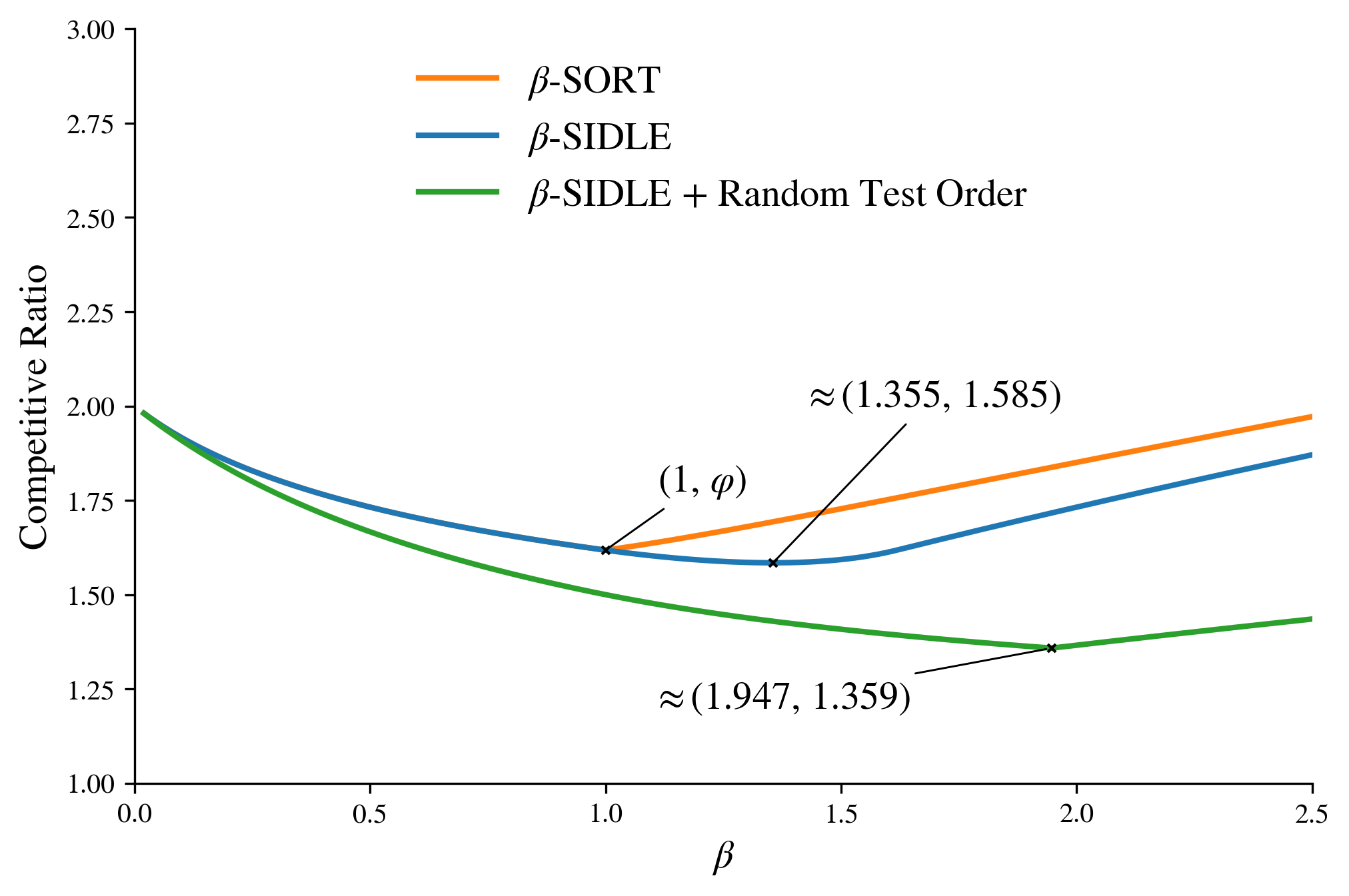}
    \caption{Comparison of exact competitive ratios for uniform testing times.}
    \label{fig: Competitive Ratios Randomized}
\end{figure}

The competitive ratios of the three proposed scheduling rules are depicted in Figure~\ref{fig: Competitive Ratios Randomized}. 
The worst-case instance structures for $\beta$-SIDLE with a random test order are achieved by instances with $p_j=\beta+\varepsilon$ for $\beta<1$, instances with $p_j\in \lbrace 0, \beta + \varepsilon\rbrace$ for $1\leq \beta \leq \beta^*\approx 1.947$, and instances with $p_j\in \lbrace 0, \beta\rbrace$ for $\beta>\beta^*$, see Appendix~\ref{appendix: parameter optimization}. For $\beta$-SIDLE, the worst-case structure is similar, but slightly different, e.g., there is a mixed-case $p_j\in\lbrace 0,\beta,\beta+\varepsilon\rbrace$ for some values of $\beta$.
For $\beta$-SORT, the branch for $\beta\in[1,\infty)$ remains the same as in Section~\ref{sec: Deterministic Algorithms}, while it collapses to the one of $\beta$-SIDLE on the branch for $\beta<1$, see Remark~\ref{remark: Uniform Testing Times}. This figure also shows that the optimal value of $\beta$ increases as the adversary becomes less powerful across these three types of algorithms, and one balances the trade-off more in favor of exploitation.
\end{remark}

In the following Section~\ref{sec: Lower Bounds for Randomized Algorithms}, we will see that there is a lower bound of $\nicefrac{4}{3}$. This means that a simple randomization step for instances with uniform testing times significantly reduces the competitive ratio.

\subsection{A Lower Bound for Randomized Algorithms}\label{sec: Lower Bounds for Randomized Algorithms}
In the stochastic counterpart of the scheduling problem with testing studied in this paper, we are given $n$ jobs to be scheduled on a single machine, with testing times $t_j$ for each job and random non-negative processing times $P_j\sim\mathcal{D} _j$ independently drawn from a distribution~$\mathcal{D}_j$ with finite support and the objective is to minimize $\E\left[\sum C_j\right]$, see e.g., \cite{levi2019scheduling,levi2024scheduling,sun2018patient}. Following Levi et al.~\cite{levi2024scheduling}, an index rule is optimal for this stochastic scheduling problem. This index rule bears similarities with $\beta$-SORT and uses indices $p_j$ for processing operations and $x_j$ being the unique solution of $t_j=\E\left[\max\lbrace x_j - P_j,0\rbrace\right]$ for the testing operation for job $j$.

We explicitly use this via the connection between randomized algorithms for adversarial models and deterministic algorithms for stochastic models, established by Yao's principle~\cite{yao1977probabilistic}. This provides an alternative, direct explanation of the lower-bound family for uniform testing times, as also given in \cite{rozhon2026test}.

\begin{lemma}
    Any randomized algorithm for scheduling with obligatory testing on a single machine to minimize $\E\left[\sum C_j\right]$ has a competitive ratio of at least $\nicefrac{4}{3}$.
\end{lemma}
\begin{proof}
    Consider the instance as in \cite{rozhon2026test} consisting of $n$ jobs with $t_j=1$ and independent $P_j$ drawn from a distribution attaining $0$ and $2$ with probability $\nicefrac{1}{2}$ each. The priority value for testing operations is $x_j=2$ as $\E\left[\max\lbrace2 - P_j,0\rbrace\right]=1$ for each job~$j$. Hence, scheduling the jobs without preemption is an optimal policy, and we compute its objective value as:
    \begin{equation}
        \E[\sum C_j]= \sum_j\E[ 1 + p_j] + \sum_{i< j} \E\left[d(i,j) + d(j,i)\right] = 2n +\frac{2n(n-1)}{2}\ ,
    \end{equation}
    and the objective value of the optimal clairvoyant policy as
    \begin{equation}
        \E[OPT]= \sum_j\E[1 + p_j] + \sum_{i< j} \E\left[1+ \min\lbrace P_i,P_j\rbrace\right] = 2n + \frac{3}{2}\cdot\frac{n(n-1)}{2} \ .
    \end{equation}
    For $n\rightarrow\infty$, we obtain a ratio approaching $\nicefrac{4}{3}$.
\end{proof}

\section{Conclusion}\label{sec: Conclusion}
In this paper, we developed an iterative instance-reduction method that is competitive-ratio-preserving, yielding highly structured instances on which the competitive ratio can be computed directly. We applied it to various index rules for scheduling with obligatory testing and showed how to precisely determine competitive ratios. 
With this, we showed that the exact competitive ratio of $1$-SORT is the Golden Ratio, improving upon the bound of $1.861$ in \cite{dogeas2024obligatory} and yielding the new best-known bound for scheduling with obligatory testing with arbitrary testing times on a single machine to minimize total completion time. 

We believe our instance-reduction technique can also be leveraged to compute exact competitive ratios for other (index) rules that rely on pairwise structures in online optimization problems in scheduling and beyond. A natural next step arising from this research is to develop randomized algorithms for obligatory testing with non-uniform test times.

\paragraph*{Declaration on AI Usage}
The ideas and proofs in this paper were developed by the author, and the paper was written by the author. The model Claude Opus~$5$ by Anthropic was used to assist with simulations, validation of calculations, presentation discussions, and parts of the parameter optimization in Appendix~\ref{appendix: parameter optimization}. 

\bibliography{bibliography}

\appendix
\section{Exact Competitive Ratio of $\beta$-SORT for $\beta\neq 1$}\label{appendix: calculation beta reduction}

\begin{proof}[Proof of Theorem~\ref{thm: worst-case families} for $\beta \neq1$]
    Let $I$ be an instance with job types $(0,\beta),(1,\beta),(1,0)$ as in Theorem~\ref{thm: Instance Reduction}. Suppose we have an instance $I$ consisting of $\mu n, \lambda n,\sigma n$ jobs of type $(0,\beta),(1,\beta),(1,0)$, where $\mu,\lambda,\sigma\geq 0$ and $\mu+\lambda+\sigma=1$. 

    Consider the case of $\beta>1$ first. Using Equation~\eqref{eq: Algorithmic cost} and Table~\ref{tab: beta > 1} to compute the pairwise delays; and further $\mu=1-\lambda-\sigma$; we obtain:
    \begin{equation}
        \begin{aligned}
         ALG &=  \frac{n^2}{2} \left( \beta\mu^2 + (2+\beta)\lambda^2 + \sigma^2+ 2\beta\mu\sigma + 2(1+\beta)\lambda(\mu+\sigma) \right) + \frac{n}{2}\left(\beta\mu + \beta\lambda + \sigma  \right) \\
         &= \frac{n^2}{2} \left( 2\lambda + \sigma^2(1-\beta) + \beta \right) +\frac{n}{2} \left(\sigma + \beta(1-\sigma) \right) \ ,
        \end{aligned}
    \end{equation}
    
    and similarly, for the optimal schedule, using Equation~\eqref{eq: Optimal cost}, Table~\ref{tab: beta > 1} and $\mu=1-\lambda-\sigma$:
    \begin{equation}
        \begin{aligned}
         OPT &= \frac{n^2}{2} \left( \beta\mu^2 + (1+\beta)\lambda^2 + \sigma^2+ 2\beta\mu\lambda + 2\sigma(\mu + \lambda) \right) +  \frac{n}{2} \left(\beta\mu + (1+\beta)\lambda + \sigma  \right) \\
         &= \frac{n^2}{2} \left( \lambda^2 + 1 + (\sigma-1)^2(\beta -1)\right) +  \frac{n}{2} \left(\lambda + \sigma + \beta(1-\sigma) \right) \ .
        \end{aligned}
    \end{equation}

    We can rewrite 
    \begin{equation}\label{eq: ratio computation beta > 1}
        \frac{ALG(I)}{OPT(I)} \leq \frac{ 2\lambda + \sigma^2(1-\beta) + \beta }{\lambda^2 + 1 + (\sigma-1)^2(\beta -1)} =: \frac{A}{B} \ .
    \end{equation}
    In addition, the ratio approaches the latter ratio as $n\rightarrow\infty$, allowing us to focus on the limiting case with $\mu, \lambda,\sigma$ fixed.
    
   If $\mu=0$, then we are done. Otherwise, we consider two cases: Either consider the type of instance $I_2$, parametrized by $\mu_2=0, \lambda_2 = \lambda +\mu, \sigma_2 = \sigma$, or instance $I_3$, parametrized by $\mu_3=0, \lambda_3 = \lambda, \sigma_3 = \sigma +\mu$ where all jobs of type $(0,\beta)$ are converted to jobs of type $(1,\beta)$ in~$I_2$ or to jobs of type $(1,0)$ in $I_3$.
   Using Equation~\eqref{eq: ratio computation beta > 1}, we have that 
    \begin{equation}\label{eq: computation case 1 beta > 1}
        \frac{A}{B} \leq \frac{A_2}{B_2} = \frac{A+ 2\mu}{B + \mu(2\lambda + \mu)} \Leftrightarrow \frac{A}{B} \leq \frac{2}{2\lambda + \mu}
    \end{equation}
    for instance $I_2$ and analogously for instance $I_3$ using $\sigma=1-\mu-\lambda$ and $\beta>1$:
    \begin{equation}\label{eq: computation case 2 beta > 1}
        \frac{A}{B} \leq \frac{A_3}{B_3} = \frac{A+ \mu(2\sigma+\mu)(1-\beta)}{B - \mu(1-\beta)(\mu + 2\sigma-2)} = \frac{A+ \mu(2- (2\lambda + \mu))(1-\beta)}{B + \mu(1-\beta)(2\lambda + \mu)} \Leftrightarrow \frac{A}{B}  \geq \frac{2}{2\lambda + \mu} -1 \ .
    \end{equation}
    Since at least one of the two conditions on $\nicefrac{A}{B}$ holds, we can find a family of instances that achieves the same competitive ratio or worse in the limit.

    For $\beta < 1$, we use the same proof strategy with slightly different terms to show that there exists an instance family such that $\sigma_2=0$ that achieves the same competitive ratio or worse.

    Using Equation~\eqref{eq: Algorithmic cost} and Table~\ref{tab: beta < 1} to compute the pairwise delays; and further $\mu=1-\lambda-\sigma$; we obtain:
    \begin{equation}
        \begin{aligned}
         ALG &=  \frac{n^2}{2} \left( \beta\mu^2 + (2+\beta)\lambda^2 + \sigma^2+ 2\sigma(2\lambda + \mu) + 2(1+\beta)\lambda\mu \right) + \frac{n}{2}\left(\beta\mu + \beta\lambda + \sigma  \right) \\
         &= \frac{n^2}{2} \left( 2\lambda +1 + (\sigma-1)^2(\beta -1) \right) +\frac{n}{2} \left(\sigma + \beta(1-\sigma) \right) \ .
        \end{aligned}
    \end{equation}
    
    For the optimal schedule, using Equation~\eqref{eq: Optimal cost}, Table~\ref{tab: beta < 1} and $\mu=1-\lambda-\sigma$, it follows:
    \begin{equation}
        \begin{aligned}
         OPT &= \frac{n^2}{2} \left(  \beta\mu^2 + (1+\beta)\lambda^2 + \sigma^2+ 2\lambda\sigma + 2\beta\mu(\sigma + \lambda) \right) +  \frac{n}{2} \left(\beta\mu + (1+\beta)\lambda + \sigma  \right) \\
         &= \frac{n^2}{2} \left( (1-\beta)(\sigma + \lambda)^2 + \beta(1+\lambda^2)\right) +  \frac{n}{2} \left(\lambda + \sigma + \beta(1-\sigma) \right) \ .
        \end{aligned}
    \end{equation}

    We can rewrite 
    \begin{equation}\label{eq: ratio computation beta < 1}
        \frac{ALG(I)}{OPT(I)} \leq \frac{ 2\lambda +1 + (\sigma-1)^2(\beta -1) }{(1-\beta)(\sigma + \lambda)^2 + \beta(1+\lambda^2)} =: \frac{A}{B} \ .
    \end{equation}
    In addition, the ratio approaches the latter ratio as $n\rightarrow\infty$, allowing us to focus on the limiting case with $\mu, \lambda,\sigma$ fixed.

    If $\sigma=0$, then we are done. Otherwise, $\sigma>0$, and we consider two cases: Either consider the type of instance $I_2$, parametrized by $\mu_2=\mu + \sigma, \lambda_2 = \lambda, \sigma_2 = 0$, or instance~$I_3$, parametrized by $\mu_3=\mu, \lambda_3 = \lambda+\sigma, \sigma_3 = 0 $ where all jobs of type $(1,0)$ are converted to jobs of type $(0,\beta)$ in~$I_2$ or to jobs of type $(1,\beta)$ in $I_3$.
    Using Equation~\eqref{eq: ratio computation beta < 1} and $\beta < 1$, we have that 
    \begin{equation}\label{eq: computation case 1 beta < 1}
        \frac{A}{B} \leq \frac{A_2}{B_2} = \frac{A - (\beta-1)\sigma(\sigma-2)}{B - (\beta-1)\sigma(-2\lambda -\sigma)} \Leftrightarrow \frac{A}{B} \geq \frac{2-\sigma}{2\lambda + \sigma} 
    \end{equation}
    for instance $I_2$ and analogously for instance $I_3$ using $\beta<1$:
    \begin{equation}\label{eq: computation case 2 beta < 1}
        \frac{A}{B} \leq \frac{A_3}{B_3} = \frac{A+\sigma(2-(\beta-1)(\sigma-2))}{B+  \beta\sigma(\sigma+2\lambda)} \Leftrightarrow \frac{A}{B} \leq \frac{2-(1-\beta)(2-\sigma)}{\beta(2\lambda + \sigma)} \ .
    \end{equation}

    We have that 
    \begin{equation}
        \frac{2-(1-\beta)(2-\sigma)}{\beta(2\lambda + \sigma)} \geq \frac{2-\sigma}{2\lambda + \sigma} \Leftrightarrow \sigma \geq 0\ .
    \end{equation}
       Hence, at least one of the two conditions on $\nicefrac{A}{B}$ holds, we can find a family of instances that achieves the same competitive ratio or worse in the limit.
\end{proof}

\section{Parameter Optimization in the Proof of Theorem~\ref{thm: Optimized Bounds}}\label{appendix: parameter optimization}
This section presents the algebraic derivations within the parameter optimization for Equation~\eqref{eq: Bound on the Ratio Randomized}, in particular:
\begin{equation}\label{eq: repeated equation}
     \max_{0\leq\mu,\lambda\leq1}\frac{1-\mu +\beta\mu^2\lambda(1-\lambda)}{1+\beta(\mu\lambda -1)^2}\ .
\end{equation}
As a first step, we show that we can restrict ourselves to parameter combinations with $\mu=1$ or $\lambda=1$. For this, we use that if $\mu<1$ and $\lambda < 1$, we can rewrite the function as a convex combination of the two with parameters $\mu_1=\mu\lambda$, $\lambda_1=1$, and $\mu_2=1$, $\lambda_2=\mu\lambda$. These are chosen such that all three parameter combinations share the same denominator, and the convex combination uses coefficients $\alpha_1=\nicefrac{(1-\mu)}{(1-\mu\lambda)}$ and $\alpha_2=\nicefrac{(\mu-\mu\lambda)}{(1-\mu\lambda)}$. We have that $\alpha_1,\alpha_2\in\left[0,1\right]$, since $0\leq \mu,\lambda < 1$, and $\alpha_1+\alpha_2=1$, and the calculation reads:
\begin{equation}\label{eq: convex combination}
     \frac{1-\mu +\beta\mu^2\lambda(1-\lambda)}{1+\beta(\mu\lambda -1)^2} = \alpha_1 \frac{1-\mu\lambda}{1+\beta(\mu\lambda-1)^2} + \alpha_2 \frac{\beta\mu\lambda(1-\mu\lambda)}{1+\beta(\mu\lambda-1)^2} \ .
\end{equation}
The left-hand side is bounded from above by the maximum of the two expressions for $\lambda_1,\mu_1$ and $\lambda_2,\mu_2$. Hence, we are left with two single-parameter optimization problems. 

For the first, we keep $\lambda=1$ fixed and vary $\mu\in\left[0,1\right]$. Differentiating $\nicefrac{(1-\mu)}{(1+\beta(\mu-1)^2)}$ and solving for stationary points yields $\beta(\mu-1)^2-1=0$ and $\mu = 1 \pm \nicefrac{1}{\sqrt{\beta}}$. For $\beta\geq1$, we hence have $\mu^*= 1- \nicefrac{1}{\sqrt{\beta}}$ and a function value of $\nicefrac{1}{(2\sqrt{\beta})}$. For $\beta<1$, we have $\mu^*=0$, since $1- \nicefrac{1}{\sqrt{\beta}}<0$, and a function value of $\nicefrac{1}{(1+\beta)}$. 

For the second, we keep $\mu=1$ fixed and vary $\lambda\in\left[0,1\right]$. Differentiating $\nicefrac{\beta\lambda(1-\lambda)}{(1+\beta(\lambda-1)^2)}$ and solving for stationary points yields $\beta\lambda^2-(1+\beta)(2\lambda-1)=0$ and $\lambda = 1 + \nicefrac{(1\pm \sqrt{1+\beta})}{\beta}$. We hence have $\lambda^*=1 + \nicefrac{(1- \sqrt{1+\beta})}{\beta}$ and a function value of $\nicefrac{(\sqrt{1+\beta}-1)}{2}$. 

Comparing the optimal expressions from the single-parameter optimizations for $\beta<1$, we have that $\nicefrac{1}{(1+\beta)}>\nicefrac{1}{2}>\nicefrac{(\sqrt{2}-1)}{2}>\nicefrac{(\sqrt{1+\beta}-1)}{2}$, hence the former term dominates on this branch. The comparison on the branch with $\beta\geq 1$ leads to an intersection at $\beta^*\approx 1.947$, the solution of $\beta^2=2\sqrt{\beta}+1$, where $\nicefrac{1}{(2\sqrt{\beta})}$ dominates for $1\leq\beta\leq \beta^*$ and $\nicefrac{(\sqrt{1+\beta}-1)}{2}$ dominates for values $\beta\geq \beta^*$. 

Overall, we obtain the expression for the competitive ratio as stated in Equation~\eqref{eq: randomized SIDLE exact competitive ratio}. The instances in the proof of Theorem~\ref{thm: Optimized Bounds} are parametrized such that there are $\mu\lambda n$ jobs of length $0$, $\mu(1-\lambda) n$ jobs of length $\beta$, and $
(1-\mu) n$ jobs of length $\beta+\varepsilon$. Hence, we get instances only consisting of jobs with $p_j=\beta+\varepsilon$ for $\beta<1$, instances with $p_j\in \lbrace 0, \beta + \varepsilon\rbrace$ for $1\leq \beta \leq \beta^*$, and instances with $p_j\in \lbrace 0, \beta\rbrace$ for $\beta>\beta^*$.
\end{document}